\documentclass[conference,11pt]{ieeeconf} 
\IEEEoverridecommandlockouts

\usepackage{amsmath}
\usepackage{color}
\usepackage{algorithm}
\usepackage{tikz}
\usepackage{graphicx}
\usepackage{authblk}
\usepackage{mathtools}
\usepackage{tikz}
\usetikzlibrary{arrows,automata}
\usepackage{amsfonts}
\usepackage{algpseudocode}
\usepackage[noadjust]{cite}
\usepackage{multirow}
\usepackage{balance}
\usepackage{url}

\title{Distributed Robust Set-Invariance for Interconnected Linear Systems}
\author{Sadra Sadraddini and Calin Belta
\thanks{
Sadra Sadraddini and Calin Belta are with the Department of Mechanical Engineering, Boston University, Boston, MA. Emails: \{sadra,cbelta\}@bu.edu. This work was partially funded by NSF grants CPS-1446151 and CMMI-1400167.
}
}

\newtheorem{define}{Definition}

\newtheorem{problem}{Problem}
\newtheorem{lemma}{Lemma}
\newtheorem{remark}{Remark}
\newtheorem{assumption}{Assumption}
\newtheorem{theorem}{Theorem}

\DeclareMathOperator{\argmin}{{\arg \min}}
\DeclareMathOperator{\argmax}{{\arg \max}}
\DeclareMathOperator{\st}{{~s.t.~}}

\begin{document}
\maketitle

\thispagestyle{empty}
\pagestyle{empty}

\begin{abstract}
We introduce a class of distributed control policies for networks of discrete-time linear systems with polytopic additive disturbances. The objective is to restrict the network-level state and controls  to user-specified polyhedral sets for all times. This problem arises in many safety-critical applications. We consider two problems. First, given a communication graph characterizing the structure of the information flow in the network, we find the optimal distributed control policy by solving a single linear program. Second, we find the sparsest communication graph required for the existence of a distributed invariance-inducing control policy. Illustrative examples, including one on platooning, are presented. 

\end{abstract}

\section{Introduction}

Centralized control of large-scale networked systems requires all the subsystems to communicate with a central coordinator, an entity which has to promptly compute control decisions for all subsystems, making centralized control impractical. Distributed control policies - where the computation and communication loads of subsystems are limited - are preferred in practice.

Designing structured controllers is difficult. A structural constraint states whether a subsystem can communicate with another subsystem. For linear interconnected systems subject to additive disturbances, which are common in applications such as formation control and energy managament, it is well-known that the problem of designing optimal stabilizing controllers (e.g., in a $\mathcal{H}_2$ or $\mathcal{H}_\infty$ sense) subject to structural constraints is NP-hard \cite{lin2013design}. Numerous methods have been proposed to design static feedback gains that respect structural constraints or lead to sparse structural requirements \cite{lin2013design,fardad2014design,fattahi2015transformation,arastoo2016closed,fazelnia2017convex,fattahiscalable}. However, since the set of stabilizing feedback gains is, in general, non-convex, the problem is computationally challenging. Moreover, a serious drawback of current methods is that they cannot take state and input constraints into account while disturbances are also present. Set-invariance specifications \cite{Blanchini:1999aa} require the system state to always remain in a user-specified safe set while inputs take values from a user-specified admissible set. In many safety-critical applications, constraint satisfaction is even more important than stabilization. For example, guaranteeing collision avoidance while respecting physical input limits is essential in vehicular platooning. The current methods of designing static feedback gains do not allow correct-by-design constraint satisfaction, and one has to test the stabilizing controller to see whether they fulfill the constraints. This process can be expensive.

In this paper, polytopic set-invariance is the main objective. It is well-known that all invariance-inducing controllers may not be described using a finite number of parameters \cite{Blanchini:1999aa}. We use the framework in \cite{rakovic2007optimized} to characterize convex sets of parameters guaranteeing set-invariance. We propose a method to impose structural constraints on the parameters. Unlike the traditional approaches discussed earlier, we require that subsystems act as relay nodes while passing information in the network. The delay of such relaying processes is taken into account in the design of the controller. In this paper, we establish the following two main results:
\begin{itemize}
\item Given a directed communication graph describing the structural constraints of the network, \emph{our method designs control policies using linear programming}. The number of constraints and variables scale polynomially with the problem size.
\item When structural constraints are not given, we find a minimal communication graph - in the sense that a weighted sum of (one-way) communication links is minimized - for which a distributed invariance-inducing control policy exists. The problem can be both solved exactly using a mixed-integer linear program or approximately solved using linear-programming relaxations. 
\end{itemize}   

Decentralized set-invariance control was considered in \cite{rakovic2010practical, Nilsson2016seperable}. Decentralized policies do not take advantage of coordination between subsystems hence they can be conservative. Convex optimization of decentralized controllers for a class of systems was established using the notion of quadratic invariance in \cite{rotkowitz2006characterization,lessard2011quadratic}. Distributed model predictive controllers (MPCs) require a distributed set-invariance property for maintaining feasibility. The authors in \cite{summers2012distributed,conte2012distributed,wang2017distributed} studied distributed MPC, but disturbances were not modeled, which significantly eases computations as Lyapunov-based approaches are used. To this end, the problem of distributed set-invariance control subject to polytopic disturbances for networks that are coupled both by dynamics and constraints remained open. This paper introduces a class of solutions to this problem.

This paper is organized as follows. The problem is stated in Sec. \ref{sec_problem}. The parametrization of invariance inducing policies is explained in Sec. \ref{sec_rci}. Computing structured control policies and designing communication graphs are covered in Sec. \ref{sec_structured} and Sec. \ref{sec_design}, respectively. Examples are presented in Sec. \ref{sec_examples}.  

\section{Definitions and Problem Formulation}
\label{sec_problem}
\subsection*{Notation}
The set of real values, non-negative real values, non-negative integers, positive integers and Boolean values are denoted by $\mathbb{R}, \mathbb{R}_+, \mathbb{N}, \mathbb{N}_+$ and $\mathbb{B}$, respectively. Cardinality of a set $S$ is denoted by $|S|$. Given a set $\mathbb{X} \subset \mathbb{R}^n$, and a matrix $A \in \mathbb{R}^{q \times n}$, we interpret $A\mathbb{X}$ as $\left\{ Ax \in \mathbb{R}^q | x \in \mathbb{R}^n \right\}$. The unit infinity-norm ball in $\mathbb{R}^n$ is denoted by $\mathcal{B}^n_\infty=\{x \in \mathbb{R}^n | \|x\|_\infty \le 1 \}.$ Given two sets $\mathbb{X}, \mathbb{Y} \subset \mathbb{R}^n$, their Minkowski sum is denoted by $\mathbb{X} \oplus \mathbb{Y}=\{x+y | x \in \mathbb{X}, y\in \mathbb{Y}\}$. Given a matrix $A$, its $(i,j)^\text{th}$ entry is denoted by $A_{[i,j]}$. The partial order relation $\le$ between two matrices of same size is interpreted entry-wise.

\subsection{Networked Control System}
A networked control system $\mathcal{S}$ is defined as a set of interconnected subsystems. The discrete-time evolution of $s \in \mathcal{S}$ is given as:
%\begin{equation}
%\label{eq_subsystem}
%x_s^+= A_s x_s + B_s u_s + w_s + \sum_{s' \in \mathcal{S}} (A_{s's} x_{s'}+B_{s's} u_{s'}),
%\end{equation}
\begin{subequations}
\label{eq_subsystem}
\begin{equation}
x_s[t+1]= A_s x_s[t] + B_s u_s[t] + w_s[t] + \sum_{s' \in \mathcal{S},s'\neq s} \zeta_{s's}[t],
\end{equation}
\begin{equation}
\zeta_{s's}[t]=A_{s's} x_{s'}[t]+B_{s's} u_{s'}[t],
\end{equation}
\end{subequations}
where $x_s[t] \in \mathbb{R}^{n_s}$, $u_s[t] \in \mathbb{R}^{m_s}$, $w_s[t]\in \mathbb{R}^{n_s}$, are the state, control and additive disturbance of system $s$, respectively, and $\zeta_{s's}[t]$ is the dynamical influence of $s'$ on $s$ at time $t \in \mathbb{N}$. Matrices $A_s \in \mathbb{R}^{n_s \times n_s}$, $B_s \in \mathbb{R}^{n_s \times m_s}$ are constant and correspond to the internal dynamics of $s$, while  $A_{s's} \in \mathbb{R}^{n_s \times n_{s'}}$, $B_{s's} \in \mathbb{R}^{n_s \times m_{s'}}$ are constant matrices characterizing the influence of $s'$ on $s$. 
Given a particular ordering of the subsystems in $\mathcal{S}$ as $(s_1,s_2,\cdots,s_N)$, where $N=|\mathcal{S}|$, the states, controls and disturbances of $\mathcal{S}$ are denoted by $x,u,$ and $w$, respectively, where:
\begin{equation}
\label{eq_compacts}
x=\left(
\begin{array}{c}
x_{s_1} \\ \vdots \\ x_{s_N} 
\end{array}
\right),
u=\left(
\begin{array}{c}
u_{s_1} \\ \vdots \\ u_{s_N} 
\end{array}
\right),
w=\left(
\begin{array}{c}
w_{s_1} \\ \vdots \\ w_{s_N} 
\end{array}
\right),
\end{equation}  
We have $x\in \mathbb{R}^n$, $u\in \mathbb{R}^m$, $w\in \mathbb{R}^n$, where
\begin{equation}
n= \sum_{s \in \mathcal{S}} {n_s}, m= \sum_{s \in \mathcal{S}} {m_s}.
\end{equation} 
The evolution of $\mathcal{S}$ is written in the following compact form: 
\begin{equation}
\label{eq_system}
x[t+1]=Ax[t]+Bu[t]+w[t],
\end{equation}
where $A\in \mathbb{R}^{n \times n}, B\in \mathbb{R}^{n \times m}$ are unambiguously constructed from \eqref{eq_subsystem} and \eqref{eq_compacts}.

\subsection{Communication Graph}
\begin{define}
A directed communication graph is defined as the tuple $\mathcal{G}=(\mathcal{S},\mathcal{L})$, where $\mathcal{S}$ (the set of subsystems) is the set of vertices and $\mathcal{L} \subseteq \mathcal{S} \times \mathcal{S}$ is a set of ordered pairs. Subsystem $s'$ is able to transmit information to subsystem $s$ if and only if $(s',s) \in \mathcal{L}$. 
\end{define}

Given $\mathcal{G}=(\mathcal{S},\mathcal{L})$, we define the $k^{th}$ power of $\mathcal{G}$ as $\mathcal{G}^k=(\mathcal{S},\mathcal{L}^k)$, $k \in \mathbb{N}_+$, such that $(s,s') \in \mathcal{L}^k$ if and only if there exists a walk from $s$ to $s'$ on $\mathcal{G}$ with length less than or equal to $k$.
Note that $G^1=G$. As a special case for $k=0$, we define $\mathcal{G}^0$ and $\mathcal{L}^0$ such that $(s,s) \in \mathcal{L}^0, \forall s \in \mathcal{S}$ (self-loops). In other words, every subsystem has access to its own information. 

\begin{assumption}
\label{assume_access}
At time $t \in \mathbb{N}$, system $s$ knows $x_{s'}[t-k+1]$ and $u_{s'}[t-k]$ if and only if $(s',s) \in \mathcal{L}^k$.   
\end{assumption}

Assumption \ref{assume_access} requires that subsystems act as relay nodes while passing information in the network. Each relay node induces a one time step delay. Assumption \ref{assume_access} is not restrictive in most applications. We require each subsystem to have some additional memory to store the history of state and controls of its own and some other subsystems. Fortunately, as it will be made clear later, only a finite (usually small) number of recent states and controls is sufficient for our purpose. In this paper, every link beyond an immediate neighbor corresponds to one unit time delay. More complex delay behavior can be accommodated in our framework by adding virtual relaying nodes (see, e.g., \cite{lamperski2015optimal}).   

\subsection{Control Objective}
We are given the following polytopes:  
\begin{subequations}
\label{eq_sets}
\begin{equation}
\mathbb{X}:= \left\{x| H_x x \le h_x \right\},
\end{equation}
\begin{equation}
\mathbb{U}:= \left\{u| H_u u \le h_u \right\},
\end{equation}
\begin{equation}
\mathbb{W}:= \left\{w| H_u w \le h_w \right\},
\end{equation}
\end{subequations}
where $H_x \in \mathbb{R}^{q_x \times n}, H_u \in \mathbb{R}^{q_u \times m}, H_w \in \mathbb{R}^{q_w \times n}$, and $h_x \in \mathbb{R}_+^{q_x}, h_u \in \mathbb{R}_+^{q_u}, h_w \in \mathbb{R}_+^{q_w}$. Note that we assume $\mathbb{X}, \mathbb{U},$ and $\mathbb{W}$ contain the origin. 
%Note that an affine transformation of coordinates can be performed to bring .    
 
\begin{define}[Centralized Policy] 
A centralized control policy is defined as $\mu^c:\mathbb{R}^n \rightarrow \mathbb{R}^m$, where 
\begin{equation}
u[t]=\mu^c(x[t]).
\end{equation} 
\end{define}

\begin{define}[Distributed Policy]
\label{define_control}
Given a networked control system $\mathcal{S}$ with communication graph $\mathcal{G}$, and a positive integer $K \ge 1$, a distributed control strategy of memory $K$ is defined as a set of functions $\mu^d:=\{\mu_s\}_{s \in \mathcal{S}}$ such that for all $t \ge K$:
\begin{equation}
\begin{array}{ll}
u_s[t]= \mu_s \Big( & \Big\{ \{ x_{s'}[t-k+1] \}_{(s',s) \in \mathcal{L}^k}, \\ 
 & \{ u_{s'}[t-k] \}_{(s',s) \in \mathcal{L}^k} \Big \}_{k \in \{1,\cdots,K\} } \Big),
\end{array}
\end{equation}
where $\mu_s: \mathbb{R}^{\eta_s} \rightarrow \mathbb{R}^{m_s}, \eta_s=\sum_{k=1}^{K} \sum_{(s',s) \in \mathcal{L}^{k-1}} n_{s'} + \sum_{k=1}^{K} \sum_{(s',s) \in \mathcal{L}^k} m_{s'}$.
\end{define}

Definition \ref{define_control} does not explain how to compute controls for $t < K$. We shift the start time to $K$ and make the following assumption. 

\begin{assumption}
\label{assume_initial}
System \eqref{eq_system} is initialized at time $t=K$ with $x[K]=0$, $x[k]=0,u[k]=0, \forall k \in [0,K-1]$.
\end{assumption}

Assuming the initial condition to be zero is restrictive but simplifies our analysis. We can drop Assumption \ref{assume_initial} at the expense of adding an initial coordination between the subsystems. The details are explained in \ref{sec_delay}. The second part of Assumption \ref{assume_initial} is not restrictive as we can always shift the start of time to $K$ and assign arbitrary values to the past.

\begin{define}[Correctness]
Given a networked control system $\mathcal{S}$ as \eqref{eq_subsystem}, \eqref{eq_system}, polytypic sets $\mathbb{X}$, $\mathbb{U}$, $\mathbb{W}$ as \eqref{eq_sets}, a communication graph $\mathcal{G}$, and a positive integer $K$, a (distributed) control $\mu$ (of memory $K$) is \emph{correct} if for all allowable sequences $w[K],w[K+1],\cdots$, $w[t] \in \mathbb{W}, \forall t\ge K$, we have $x[t] \in \mathbb{X}$ and $u[t] \in \mathbb{U}, \forall t \in \mathbb{N}$.
\end{define}

\begin{define}[Margin of Correctness]
Given a correct control policy $\mu$, the margin of correctness $\rho^* \in [0,1]$ is defined as the maximum value of $\rho$ for which $\mu$ remains correct when $\mathbb{X} \gets (1-\rho) \mathbb{X}$ and $\mathbb{U} \gets (1-\rho) \mathbb{U}$.       
\end{define}

The margin of correctness has a straightforward interpretation. If $\rho^*=0$, it implies that correctness is lost if $\mathbb{X}$ or $\mathbb{U}$ are shrunk around the origin. If $\rho^*=1$, it indicates that the state and controls can be always zero, which essentially requires $\mathbb{W}=\{0\}$. More complicated definitions for margin of correctness are also possible. For instance, one may consider different scaling variables for components in $\mathbb{X}$ and $\mathbb{U}$ and define the margin as a weighted sum of them.

\subsection{Problem Statement}
We formulate two problems. In both, we are given a networked control system $\mathcal{S}$ as \eqref{eq_subsystem}, \eqref{eq_system}, polytypic sets $\mathbb{X}$, $\mathbb{U}$, $\mathbb{W}$ as \eqref{eq_sets}, and a positive integer $K$. In practice, $K$ is a design parameter which determines the complexity of the controller. We usually start from small values of $K$ and make it larger until feasibility/satisfactory performance is reached.   

\begin{problem}[Optimal Strategy Design]
Given a communication graph $\mathcal{G}$, design a correct distributed control policy $\mu$ of memory $K$ with the maximum margin of correctness $\rho^*$.  
\end{problem}

\begin{problem}[Optimal Graph Design]
Find a communication graph $\mathcal{G}=(\mathcal{S},\mathcal{L})$ for which a correct control policy exists such that the following cost function is minimized:
 \begin{equation}
 J (\mathcal{G})= \sum_{s \in \mathcal{S}}\sum_{s' \in \mathcal{S}} c_{s's} \mathcal{I}\left((s',s) \in \mathcal{L} \right),
 \end{equation}
 where $c_{s's} \in \mathbb{R}_+^n$ is the cost of establishment of one-way communication link from $s'$ to $s$, and $\mathcal{I}$ is the indicator function that designates $1$ (respectively, $0$) if its argument is true (respectively, false). 
 
\end{problem}

\section{Parameterized Set-Invariance}
\label{sec_rci}
In this section, we present the family of parameterized controllers in \cite{rakovic2007optimized}. We do not, yet, impose structural constraints. The key idea of this paper is outlined in Sec. \ref{sec_delay}, where we show that a memoryless piecewise affine invariance-inducing control policy can be converted to a linear controller with memory, paving the path to impose structural requirements in Sec. \ref{sec_structured}.

\subsection{Convex Parameterization}
\begin{lemma}
\label{eq_lemma_sasha}
\cite{rakovic2007optimized}
Let $\Theta:=\left(\theta_0,\theta_1,\cdots,\theta_{K-1} \right)$, where $\theta_{k} \in \mathbb{R}^{m \times n}, k=0,\cdots,K-1$, be a $m \times {nK}$ matrix of parameters such that the following condition holds:
\begin{equation}
\label{eq_condition}
A^{K} + A^{K-1} B \theta_0 + \cdots + A B\theta_{K-2} + B \theta_{K-1}=0.
\end{equation}
Define the following set:
\begin{equation}
\label{eq_omega}
\begin{array}{rl}
\Omega_{\Theta}  := & (A^{K-1} + A^{K-2} B \theta_0 + \cdots + B\theta_{K-2}) \mathbb{W} \\
 & \oplus(A^{K-2} + A^{K-3} B \theta_0 + \cdots + B\theta_{K-3}) \mathbb{W} \\
 & \oplus \cdots \oplus (A+B\theta_{0}) \mathbb{W} \oplus \mathbb{W} \\
\end{array}
\end{equation}
Then there exists $\mu: \mathbb{R}^n \rightarrow \mathbb{R}^m$ such that 
$$\forall x \in \Omega_{\Theta}, \{ Ax+B \mu(x)\} \oplus \mathbb{W} \subseteq \Omega_{\Theta}.$$ 
\end{lemma}
\begin{proof}
For all $x\in \Omega_{\Theta}$, there exists $w^{K-1},w^{K-2},\cdots,w^{0} \in \mathbb{W}$ such that
\begin{equation}
\label{eq_x}
\begin{array}{rl}
x= & (A^{K-1} + A^{K-2} B \theta_0 + \cdots + B\theta_{K-2}) w^{K-1}_x + \\
 & (A^{K-2} + A^{K-3} B \theta_0 + \cdots + B\theta_{K-3}) w^{K-2}_x \\
 & + \cdots + (A+B\theta_{0}) w_x^{1} + w_x^0.
\end{array}
\end{equation}
Now let the $\mu^c(x)$ be the following control input: 
\begin{equation}
\label{eq_control}
\mu^c(x)= \theta_{K-1} w_x^{K-1} + \theta_{K-2} w_x^{K-2} + \cdots + \theta_{0} w_x^{0}.   
\end{equation}
Denote the new disturbance hitting system by $w^{+}$. The subsequent state is 
\begin{equation}
\label{eq_evolution}
\begin{array}{rl}
x^+= &  Ax+B\mu^c(x)+w^+= 
\\ & (A^{K} + A^{K-1} B \theta_0 + \cdots + B\theta_{K-1}) w_x^{K-1} \\
 & +(A^{K-1} + A^{K-3} B \theta_0 + \cdots + B\theta_{K-2}) w_x^{K-2} \\
 & + \cdots + (A+B\theta_{0}) w_x^{0} + w^+.
\end{array}
\end{equation}
Substituting \eqref{eq_condition} in \eqref{eq_evolution} results in:
\begin{equation}
\label{eq_induction}
\begin{array}{rl}
x^+= & (A^{K-1} + A^{K-2} B \theta_0 + \cdots + B\theta_{K-2}) w_x^{K-2} \\
 & +(A^{K-2} + A^{K-3} B \theta_0 + \cdots + B\theta_{K-2}) w_x^{K-3} \\
 & + \cdots + (A+B\theta_{0}) w_x^{0} + w^+.
\end{array}
\end{equation} 
A quick inspection of \eqref{eq_omega} and \eqref{eq_induction} verifies $x^+ \in \Omega_{\Theta}$.
\end{proof}
Notice that \eqref{eq_condition} is not restrictive since $\Theta$ is non-empty for controllable $(A,B)$ and $K$ greater than its controllability index. The set $\Omega_\Theta$ is a \emph{robust control invariant} (RCI) set. Following \eqref{eq_control}, the set of all possible controls is
\begin{equation}
\Psi_\Theta := \bigoplus_{k=0}^{K-1} \theta_k \mathbb{W}.
\end{equation}
In order to have a correct control policy, we require $\Omega_\Theta \subseteq \mathbb{X}$ and $\Psi_\Theta \subseteq \mathbb{U}$. 
\begin{lemma}\cite{rakovic2007optimized}
The set of parameters $\Theta, \alpha \in \mathbb{R}_+$, for which $\Omega_\Theta \subseteq \alpha \mathbb{X}$ and $\Psi_\Theta \subseteq \alpha \mathbb{U}$ is convex.
\end{lemma}
\begin{proof} The proof is based on an extension of Farkas's Lemma: Given sets $\mathbb{S}=\{ s \in \mathbb{R}^n| H_{s} s \le h_{s} \},$ and $Y \subset \{ y \in \mathbb{R}^q| H_{y} y \le h_{y} \}, $ and matrices $L_i \in \mathbb{R}^{q\times n}$, $i=1,\cdots,\kappa$, then $
\bigoplus_{i=1}^{\kappa} L_i \mathbb{S} \subseteq \alpha \mathbb{Y},
$ is equivalent to the following set of constraints: 
\begin{equation}
Z_i H_{s} = H_y L_i, i=1,\cdots,\kappa, \sum_{i=1}^{\kappa} Z_i h_s \le \alpha h_y, 
\end{equation}
where $Z_i \in \mathbb{R}^{q\times n}$, $i=1,\cdots,\kappa$, are appropriately sized matrices with non-negative entries. Convexity follows from the linearity of the constraints. 
\end{proof}

\begin{remark}
The family of RCI sets introduced in \cite{rakovic2007optimized} is not necessarily equivalent to the set of all RCI sets. In particular, there is no guarantee that one can find the maximal RCI set using this approach, which is not surprising as the problem of finding the maximal RCI set is not decidable, in general \cite{Blanchini:1999aa}. Nevertheless, the set of RCI sets in \cite{rakovic2007optimized} is quite rich as they are generated using piecewise affine feedback laws rather than the much more limited class of linear control laws.
\end{remark}

\subsection{Centralized Policy}

Given $\Theta$, the map from $x$ to $u$ requires finding values for $w_x^{K-1},\cdots,w_x^0$, subject to \eqref{eq_x}. This is accomplished by solving a linear/quadratic program. A centralized policy can be constructed as: 
\begin{equation}
\label{eq_pwa}
\begin{array}{ccl}
u[t]=\mu^c(x[t])=&  \displaystyle \underset{u}{\argmin} & \pi(u) \\
& \st & \eqref{eq_x},\eqref{eq_control}, w^k_x \in \mathbb{W},\\
& & k=0,\cdots,K-1,
\end{array}
\end{equation}
where $\pi: \mathbb{R}^m \rightarrow \mathbb{R}$ is a user-defined convex linear/quadratic cost function. It is well-known that $\mu^c$ becomes a piecewise affine function. 
\begin{remark}
\label{remark_admm}
An alternative way to make a set-invariance control policy distributed is using the state-of-the-art distributed convex optimization techniques to solve \eqref{eq_pwa}, such as alternating direction method of multipliers (ADMM) \cite{boyd2011distributed}. However, a substantial communication and computation effort is required to perform the iterations in ADMM \cite{summers2012distributed}. 
\end{remark}

\subsection{Linear Delay Policy}
\label{sec_delay}
The following result states that any memoryless piecewise affine policy obtained from \eqref{eq_pwa} can be converted into a linear policy with memory $K$.   
\begin{theorem}
Let $\Theta$ such that $\Omega_\Theta \subseteq \mathbb{X}$ and $\Psi_\Theta \subseteq \mathbb{U}$. Then a control policy in which control decisions for $t \ge K$ are given as:
\begin{equation}
\label{eq_delay_policy}
\begin{array}{ll}
u[t] = & \theta_{K-1} w[t-K]+ \theta_{K-2} w[t-K+1] \\
  & + \cdots + \theta_{0} w[t-1] \\
 \end{array}
\end{equation}
is correct if the following condition holds:
\begin{equation}
\label{eq_xk}
\begin{array}{ll}
x[K]= & (A^{K-1} + A^{K-2} B \theta_0 + \cdots + B\theta_{K-2}) w[0] + \\
 & (A^{K-2} + A^{K-3} B \theta_0 + \cdots + B\theta_{K-3})w[1] \\
 & + \cdots + (A+B\theta_{0}) w[K-2] + w[K-1].
\end{array}
\end{equation}

\end{theorem}
\begin{proof}
We prove correctness by showing that 
\begin{equation}
\begin{array}{l}
x[t]=  (A^{K-1} + A^{K-2} B \theta_0 + \cdots + B\theta_{K-2}) w[t-K] \\
+ (A^{K-2} + A^{K-3} B \theta_0 + \cdots + B\theta_{K-3})w[t-K+1] \\
 + \cdots + (A+B\theta_{0}) w[t-2] + w[t-1].
\end{array}
\end{equation}
for all $t \ge K$. We prove by induction. For $t=K$, the statement is assumed true as \eqref{eq_xk}. We prove the inductive step using \eqref{eq_condition} to arrive in:
\begin{equation*}
\begin{array}{l}
x[t+1] = Ax[t]+Bu[t]+w[t]  \\
= (A^{K-1} + A^{K-2} B \theta_0 + \cdots + B\theta_{K-2}) w[t-K+1] \\
+ (A^{K-2} + A^{K-3} B \theta_0 + \cdots + B\theta_{K-3})w[t-K+2] \\
 + \cdots + (A+B\theta_{0}) w[t-1] + w[t].
\end{array}
\end{equation*}
It follows from \eqref{eq_omega} and \eqref{eq_control} that $x[t] \in \Omega_\Theta, u[t] \in \Psi_\Theta, \forall t \ge K$, and the proof is complete.  
\end{proof}
Eq. \eqref{eq_delay_policy} is a linear policy based on disturbances. However, disturbances are not assumed to be directly measurable. Using \eqref{eq_system}, we can replace disturbances by state and controls to obtain a more useful form of \eqref{eq_delay_policy}:
\begin{equation}
\label{eq_linear}
\begin{array}{ll}
u[t] = & (-\theta_{K-1}A) x[t-K] \\
& + (\theta_{K-1}-\theta_{K-2}A) x[t-K+1] \\
& + (\theta_{K-2}-\theta_{K-3}A) x[t-K+2] \\
& + \cdots + (\theta_{1}-\theta_{0}A) x[t-1] + \theta_0 x[t] \\
& - \big( \theta_{K-1} B u[t-K] + \theta_{K-2} B u[t-K+1] \\
& + \cdots + \theta_{0} B u[t-1]\big). \\
\end{array}
\end{equation}
Since all zeros is a trivial solution to \eqref{eq_xk}, a simple way to make \eqref{eq_xk} true is holding Assumption \ref{assume_initial}. For any initial condition $x[K] \in \Omega_\Theta$, we can find hypothetical values for $w[0],w[1],\cdots,w[K-1]$ such that \eqref{eq_xk} holds by solving a linear program. However, solving such a linear program may require a central entity. We may use distributed linear program solvers (see Remark \ref{remark_admm}) to accomplish this task. Therefore, Assumption \ref{assume_initial} is relaxable given arbitrary initial conditions, as long as they lie in the RCI set. Note that if the initial condition is outside of the (maximal) RCI set, satisfying the set-invariance objective is impossible. 
%The problem of steering the initial state into the RCI set requires reachability control which is not the focus of this paper.  
\section{Control with Structural Constraints}
\label{sec_structured}
Here we provide the solution to Problem 1.
We impose structural requirements on \eqref{eq_linear} based on Assumption \ref{assume_access}. We define the following sets of matrices:
\begin{subequations}
\label{eq_ssets}
\begin{equation}
\begin{array}{ll}
\mathbb{S}^x\left((\mathcal{S},\mathcal{L})\right):=& \Big \{ G \in \mathbb{R}^{m\times n} \big|
 u_{[i]} \in s, x_{[j]} \in s',
 \\ &  (s',s) \not \in \mathcal{L} \Rightarrow G_{[i,j]}=0 \Big\},
\end{array}
\end{equation}
\begin{equation}
\begin{array}{ll}
\mathbb{S}^u\left((\mathcal{S},\mathcal{L})\right):=& \Big \{ G \in \mathbb{R}^{m\times n} \big|
 u_{[i]} \in s, u_{[j]} \in s',
 \\ &  (s',s) \not \in \mathcal{L} \Rightarrow G_{[i,j]}=0 \Big\},
\end{array}
\end{equation}
\end{subequations}
where $x_{[i]} \in s$ ($u_{[i]} \in s$) is interpreted as whether $i$'th component of $x$ ($u$) belongs to subsystem $s$. Sets in \eqref{eq_ssets} are convex. The coefficients that relate a component of $u[t]$ in \eqref{eq_linear} to a component of $x[t-k+1],u[t-K], k=1,\cdots,K$, has to be zero if it violates Assumption \ref{assume_access}, which is formally stated as follows:
\begin{equation}
\label{eq_structural}
\left\{
\begin{array}{l}
 \theta_{K-1}A \in \mathbb{S}^x(\mathcal{G}^{K+1}) \\
 \theta_{K-1}-\theta_{K-2}A \in \mathbb{S}^x(\mathcal{G}^{K}) \\
 \vdots \\
 \theta_{1}-\theta_{0}A \in \mathbb{S}^x(\mathcal{G}^{2}) \\
 \theta_{0} \in \mathbb{S}^x(\mathcal{G}^{1})
\end{array}
\right.
,\left\{
\begin{array}{l}
 \theta_{K-1}B \in \mathbb{S}^u(\mathcal{G}^{K}) \\
\vdots \\
 \theta_{0}B \in \mathbb{S}^u(\mathcal{G}^{1})
\end{array}
\right.
\end{equation}
Finally, the solution to Problem 1 is found by solving the following linear program:
\begin{equation}
\label{eq_solution1}
\begin{array}{cll}
\{\rho^*,\Theta^*\}=&  \displaystyle \underset{\Theta,\rho}{\argmax} & \rho \\
& \text{subject to} & \eqref{eq_condition},\eqref{eq_structural},\\
&& \Omega_\Theta \subseteq (1-\rho) \mathbb{X}, \\
& & \Psi_\Theta \subseteq (1-\rho) \mathbb{U}.
\end{array}
\end{equation}
\subsection*{Complexity}
The number of variables and constraints in \eqref{eq_solution1} scales linearly with respect to $K$, $n$, $m$, and the number of rows in $\mathbb{X}, \mathbb{U}, \mathbb{W}$. In practice, representation complexity of sets in \eqref{eq_sets} scale polynomially in $n$ and $m$, while the exact degree of growth depends on the application. Thus, taking the complexity of the interior-point linear programming methods into account, the overall complexity of our solution to Problem 1 increase polynomially with respect to the problem size.
%, with the degree of roughly 5-8, depending on the application and representations of $\mathbb{X}, \mathbb{U}$, and $\mathbb{W}$.      

\section{Structure Design}
\label{sec_design}
Here we provide the solution to Problem 2.
\subsection{Binary encoding}
We need to make binary decisions on whether $s'$ is connected to $s$ on $\mathcal{G}$. This task is captured by introducing binary $N(N-1)$ binary variables $b_{s's}\in \{0,1\}, s,s' \in \mathcal{S}$. Note that $b_{ss}=1, \forall s \in \mathcal{S}$. We define the adjacency matrix of $\mathcal{G}$ as $\mathcal{B}(\mathcal{G}) \in \mathbb{B}^{N \times N}$ such that $\mathcal{B}(\mathcal{G})_{[s's]}:=b_{s's}$. The following property follows from the basic properties of powers of adjacency matrix in graph theory:
\begin{equation}
\label{eq_powers}
\mathcal{B}(\mathcal{G}^k)=\sum_{p=1}^{k} \mathcal{B}(\mathcal{G}^p),
\end{equation}
where both summation and multiplication are defined in a Boolean sense, i.e., for $b_1,b_2 \in \mathbb{B}$ we have $b_1b_2=b_1\wedge b_2$ and $b_1+b_2=b_1 \vee b_2$. (e.g., $1+1=1$). 
%For example, for the graphs shown in Fig. \ref{fig_graph} we have 
%\begin{equation*}
%\tiny
%\mathcal{B}(\mathcal{G}^2)=\left(
%\begin{array}{cccc}
%1 & 1 & 1 & 0 \\
%0 & 1 & 1 & 1 \\
%0 & 0 & 1 & 1 \\
%0 & 0 & 0 & 1 \\
%\end{array}
%\right),
%\mathcal{B}(\mathcal{G}^3)=\left(
%\begin{array}{cccc}
%1 & 1 & 1 & 1 \\
%0 & 1 & 1 & 1 \\
%0 & 0 & 1 & 1 \\
%0 & 0 & 0 & 1 \\
%\end{array}
%\right).
%\end{equation*}

Given $\mathcal{B} \in \mathbb{B}^{N\times N}$, we define $\mathcal{B}^x \in \mathbb{B}^{m \times n}$ and $\mathcal{B}^u \in \mathbb{B}^{m \times m}$ such that
\begin{subequations}
\begin{equation}
\mathcal{B}^x((\mathcal{S},\mathcal{L}))_{[i,j]}=b_{s's}, u_{[i]} \in s_{i}, x_{[j]} \in s_{[j]},
\end{equation}
\begin{equation}
\mathcal{B}^u((\mathcal{S},\mathcal{L}))_{[i,j]}=b_{s's}, u_{[i]} \in s_{i}, u_{[j]} \in s_{[j]}.
\end{equation}
\end{subequations}
Given a matrix $C\in \mathbb{R}^{m \times n}$, and $k \in \mathbb{N}_+$ the following relation holds:
\begin{equation}
\label{eq_binary_sets}
C \in \mathbb{S}^{z}(\mathcal{G}^k) \Leftrightarrow -M \mathcal{B}^z(\mathcal{G})\le C \le -M \mathcal{B}^z(\mathcal{G}), 
\end{equation} 
where $z=x,u$, and $M$ is a sufficiently large positive number that is greater than $\max_{i,j} |C_{[i,j]]}|$. The constraints in \eqref{eq_powers}, \eqref{eq_binary_sets}, are mixed binary-linear constraints. We need only to declare the entries in $\mathcal{B}(\mathcal{G})$ as binaries - there are $N(N-1)$ of them - and all the other relations in \eqref{eq_powers} are captured using continuous auxiliary variables declared over $[0,1]$ - which constraints enforce them take values from $\{0,1\}$. Encoding Boolean functions using mixed binary-linear constraints is a standard procedure (see, e.g., \cite{Bemporad1999}) and the details are not presented here.   

\subsection{Graph Optimization}
\label{sec_gop}
Finally, we find the optimal communication graph $\mathcal{G}^*$ - the solution to Problem 2 - as the following mixed-integer linear program (MILP): 

\begin{equation}
\label{eq_solution2}
\begin{array}{cll}
\left\{\Theta^*, \mathcal{B}(\mathcal{G}^*) \right\}=&  \displaystyle \underset{\Theta, b_{s's},s,s' \in \mathcal{S}}{\argmin} & \sum_{s,s' \in \mathcal{S}} b_{s's}  \\
& \text{subject to} & \eqref{eq_condition},\eqref{eq_structural},\eqref{eq_binary_sets},\\
& & \Omega_\Theta \subseteq \mathbb{X},\Psi_\Theta \subseteq \mathbb{U}.
\end{array}
\end{equation}
Note that \eqref{eq_solution2} provides both a communication graph and a corresponding RCI set and distributed control policy parameterized by $\Theta^*$. Note that we can combine Problem 1 and Problem 2 by adding an additional term to the cost function in \eqref{eq_solution2} to promote greater margin of correctness. The trade-off between sparser graph and greater margin of correctness can be controlled by designating weights to the corresponding terms.  

\subsection*{Complexity}

Unlike \eqref{eq_solution1}, solving \eqref{eq_solution2} is NP-hard. MILP solvers use branch and bound techniques to explore optimal solutions by solving linear-program relaxations of the original problem. In order to find suboptimal but (arbitrary) faster solutions, a simple approach is terminating the MILP solver early -  after it has an incumbent feasible solution. (see Fig. \ref{fig_solver} in the examples). 

%\section{Discussion}
%\label{sec_discuss}

\section{Examples}
\label{sec_examples}

We have developed a python script that solves Problem 1 and Problem 2 given the system, specification and relevant parameters by the user. This script, as well as the codes for the example below, are publicly available in \texttt{github.com/sadraddini/distinct}.
\subsection{Coupled Double Integrators}

We consider $N=5$ double-integrators with state and control couplings. Fo all $s,s' \in \mathcal{S}$, we assign the following values to \eqref{eq_subsystem}:
\begin{equation*}
A_s=\left(
\begin{array}{cc}
1+\epsilon & 1 \\
-\epsilon & 1+\epsilon 
\end{array}
\right),
A_{s's}=\left(
\begin{array}{cc}
\epsilon & -\epsilon \\
-\epsilon & \epsilon 
\end{array}
\right),
\end{equation*}
\begin{equation*}
B_s=\left(
\begin{array}{c}
0 \\
1
\end{array}
\right),
B_{s's}=\left(
\begin{array}{c}
-\epsilon \\
\epsilon  
\end{array}
\right),
\end{equation*}
where $\epsilon$ is a constant characterizing the degree of coupling. We explore the behavior of solutions versus multiple values of $\epsilon>0$. For any $\epsilon>0$, at least one of the eigenvalues of $A$ lies out of the unit circle. Thus $A$ is unstable. We let $\mathbb{X}= \mathcal{B}^{10}_\infty, \mathbb{U}= 2\mathcal{B}^{5}_\infty, \mathbb{W}= \eta\mathcal{B}^{10}_\infty$, where $\eta$ is also a constant we vary in this example.    

\subsubsection{Structured Control} We solve Problem 1. We are given a communication graph that is  circular, as illustrated in Fig.  \ref{fig_given_graph}. We consider both the directed and the undirected case. The results for various values of $K$, $\eta$, $\epsilon$, are shown in Table \ref{table_margin}. As expected, the margins are smaller when coupling and disturbances are greater, and communications are directed. Also, higher values of $K$ usually correspond to better performance. For $K<6$, we could not find a solution for the directed graph. Projections of the RCI set and sample trajectories are illustrated in Fig. \ref{fig_rci}. It is observed that the undirected circular communication graph is able to keep the state closer to zero, while the RCI set and trajectories of the directed graph get closer to the boundaries of $\mathbb{X}$. All the computations in Table \ref{table_margin} were performed using Gurobi linear program solver on a dual core 3GHz MacBook Pro. The computation times were all less than a second.

\begin{figure}
\vspace{10pt}
\centering
~~~
\begin{tikzpicture}[xscale=1.3,yscale=1.3]
\tikzset{vertex/.style = {shape=circle,draw,minimum size=1pt, inner sep=0pt}}
\begin{scope}[edge/.style = {->,> = latex',line width=0.7pt},node distance=2pt]
\node[vertex] (1) at  (0,1.3) {$s_1$};
\node[vertex] (2) at  (-0.7,0.8) {$s_2$};
\node[vertex] (3) at  (-0.5,0) {$s_3$};
\node[vertex] (4) at  (0.5,0) {$s_4$};
\node[vertex] (5) at  (0.7,0.8) {$s_5$};
\draw[edge] (1) to (2);
\draw[edge] (2) to (3);
\draw[edge] (3) to (4);
\draw[edge] (4) to (5);
\draw[edge] (5) to (1);
\end{scope}
\end{tikzpicture}
~~~~
\begin{tikzpicture}[xscale=1.3,yscale=1.3]
\tikzset{vertex/.style = {shape=circle,draw,minimum size=1pt, inner sep=0pt}}
\begin{scope}[edge/.style = {<->,> = latex',line width=0.7pt},node distance=2pt]
\node[vertex] (1) at  (0,1.3) {$s_1$};
\node[vertex] (2) at  (-0.7,0.8) {$s_2$};
\node[vertex] (3) at  (-0.5,0) {$s_3$};
\node[vertex] (4) at  (0.5,0) {$s_4$};
\node[vertex] (5) at  (0.7,0.8) {$s_5$};
\draw[edge] (1) to (2);
\draw[edge] (2) to (3);
\draw[edge] (3) to (4);
\draw[edge] (4) to (5);
\draw[edge] (5) to (1);
\end{scope}
\end{tikzpicture}
\caption{Circular graphs: [Left]: directed [Right]: undirected}
\label{fig_given_graph}
\end{figure}
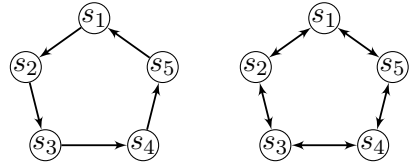

\begin{table}[t]
\vspace{0.07in}
\centering
\caption{Margin of Correctness for Graphs in Fig. \ref{fig_given_graph} }
\resizebox{0.34\textwidth}{!}{%
\begin{tabular}{|c||c|c|c|c|}
\hline
\multirow{2}{*}{$K$} & \multirow{2}{*}{$\eta$} & \multirow{2}{*}{$\epsilon$} & \multicolumn{2}{|c|}{$\rho^*$} \\ \cline{4-5}
 &  &  & Directed & Undirected \\

\hline 
 $6$ & $0.05$ & $0.05$ & $0.27$ & $0.75$ \\
 $6$ & $0.1$ & $0.1$ & Infeasible & $0.33$ \\
 $6$ & $0.1$ & $0.01$ & $0.02$ & $0.58$ \\
 $4$ & $0.05$ & $0.01$ & Infeasible & $0.79$ \\
 $4$ & $0.05$ & $0.05$ & Infeasible & $0.75$ \\
 $6$ & $0.05$ & $0.01$ & $0.51$ & $0.79$ \\
\hline
\end{tabular}
}
\label{table_margin}
\end{table}

\begin{figure}[t]
\centering
\includegraphics[width=0.21\textwidth]{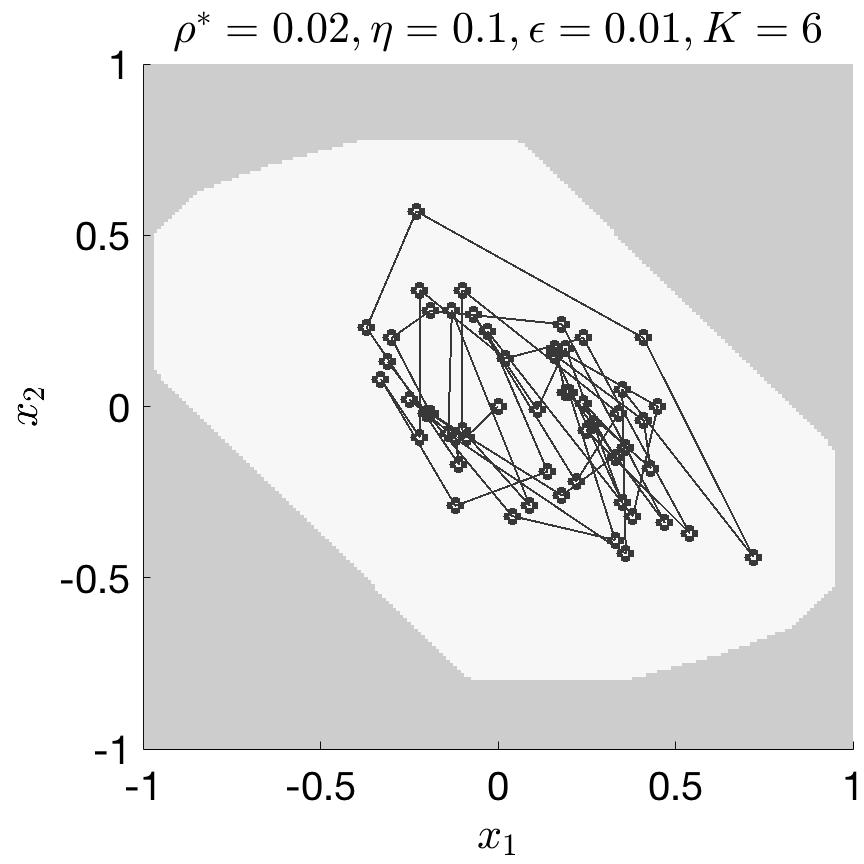}
\includegraphics[width=0.21\textwidth]{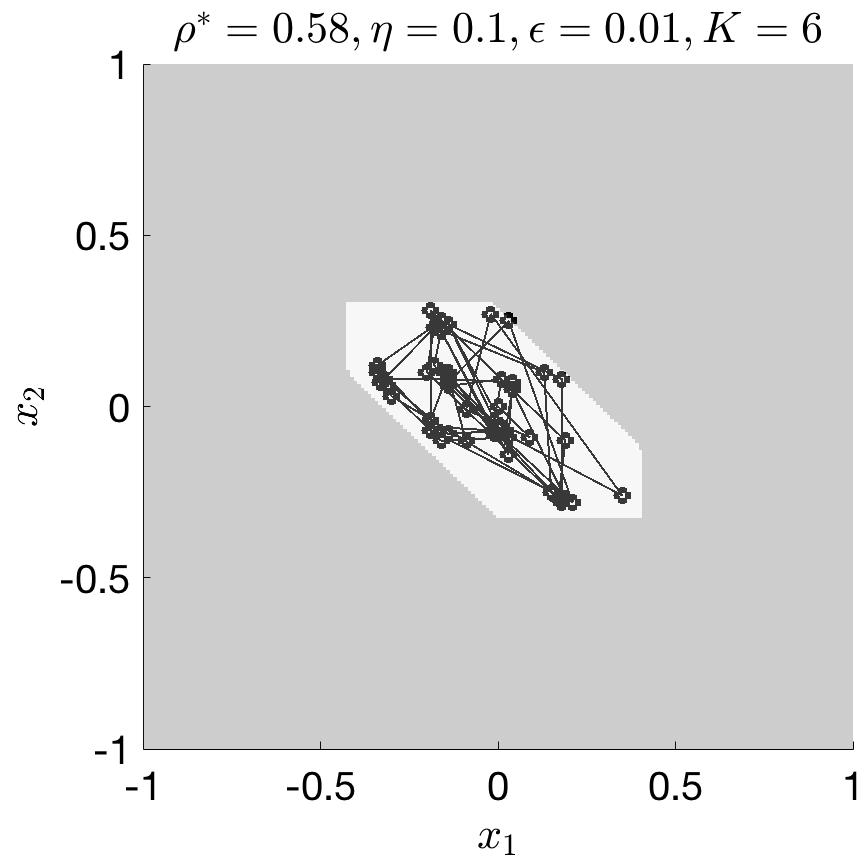}
\caption{Projection of the RCI set on the state-space of $s_1$ and a sample trajectory of 60 time steps. RCI sets correspond to [Left]: directed [Right]: undirected graphs in Fig. \ref{fig_given_graph}.}
\label{fig_rci}
\end{figure}

\subsubsection{Graph Design}
We solve Problem 2. We let $c_{s's}=1,\forall s,s' \in \mathcal{S}, s \neq s'$. Various optimal graphs corresponding to different values are shown in Fig. \ref{fig_various}. We often obtained graphs that were strongly connected. However, in case the couplings are sufficiently weak, fully decentralized solutions were found, as shown in the null graph in Fig. \ref{fig_various} (f). The computations were performed using Gurobi MILP solver on a 3GHz dual core MacBook Pro. As discussed in Sec. \ref{sec_gop}, MILP solvers explore solutions using branch and bound techniques. Two instances of the best incumbent solution versus time is shown in Fig. \ref{fig_solver}. We can obtain suboptimal solutions by early manual termination.  

\begin{figure}[t]
\centering
\vspace{10pt}
\begin{tikzpicture}[xscale=1.3,yscale=1.3]
\tikzset{vertex/.style = {shape=circle,draw,minimum size=1pt, inner sep=0pt}}
\begin{scope}[edge/.style = {->,> = latex',line width=0.7pt},node distance=2pt]
\node[vertex] (1) at  (0,1.3) {$s_1$};
\node[vertex] (2) at  (-0.7,0.8) {$s_2$};
\node[vertex] (3) at  (-0.5,0) {$s_3$};
\node[vertex] (4) at  (0.5,0) {$s_4$};
\node[vertex] (5) at  (0.7,0.8) {$s_5$};
\draw[edge] (1) to (2);
\draw[edge] (1) to (3);
\draw[edge] (1) to (4);
\draw[edge] (1) to (5);

\draw[edge] (2) to (1);
\draw[edge] (2) to (4);
\draw[edge] (2) to (5);

\draw[edge] (3) to (1);
\draw[edge] (3) to (4);

\draw[edge] (4) to (1);
\draw[edge] (4) to (2);
\draw[edge] (4) to (3);
\draw[edge] (4) to (5);

\draw[edge] (5) to (1);
\draw[edge] (5) to (4);

\node at (0,-0.2) {\tiny (a) $K=3$};
\node at (0,-0.45) {\tiny $\eta=0.2,\epsilon=0.06$};
\node at (0,-0.7) {\tiny $J(\mathcal{G}^*)=14$, CT=82s};
\end{scope}
\end{tikzpicture}
\begin{tikzpicture}[xscale=1.3,yscale=1.3]
\tikzset{vertex/.style = {shape=circle,draw,minimum size=1pt, inner sep=0pt}}
\begin{scope}[edge/.style = {->,> = latex',line width=0.7pt},node distance=2pt]
\node[vertex] (1) at  (0,1.3) {$s_1$};
\node[vertex] (2) at  (-0.7,0.8) {$s_2$};
\node[vertex] (3) at  (-0.5,0) {$s_3$};
\node[vertex] (4) at  (0.5,0) {$s_4$};
\node[vertex] (5) at  (0.7,0.8) {$s_5$};
\draw[edge] (1) to (2);
\draw[edge] (2) to (1);
\draw[edge] (2) to (3);
\draw[edge] (2) to (4);
\draw[edge] (3) to (2);
\draw[edge] (3) to (4);
\draw[edge] (4) to (5);
\draw[edge] (4) to (3);
\draw[edge] (5) to (3);
\node at (0,-0.2) {\tiny (b) $K=6$};
\node at (0,-0.45) {\tiny $\eta=0.2,\epsilon=0.1$};
\node at (0,-0.7) {\tiny $J(\mathcal{G}^*)=9$, CT=321s};
\end{scope}
\end{tikzpicture}
\begin{tikzpicture}[xscale=1.3,yscale=1.3]
\tikzset{vertex/.style = {shape=circle,draw,minimum size=1pt, inner sep=0pt}}
\begin{scope}[edge/.style = {->,> = latex',line width=0.7pt},node distance=2pt]
\node[vertex] (1) at  (0,1.3) {$s_1$};
\node[vertex] (2) at  (-0.7,0.8) {$s_2$};
\node[vertex] (3) at  (-0.5,0) {$s_3$};
\node[vertex] (4) at  (0.5,0) {$s_4$};
\node[vertex] (5) at  (0.7,0.8) {$s_5$};
\draw[edge] (1) to (2);
\draw[edge] (2) to (1);
\draw[edge] (2) to (3);
\draw[edge] (2) to (4);
\draw[edge] (2) to (5);
\draw[edge] (3) to (2);
\draw[edge] (4) to (2);
\draw[edge] (5) to (2);
\node at (0,-0.2) {\tiny (c) $K=4$};
\node at (0,-0.45) {\tiny $\eta=0.2,\epsilon=0.05$};
\node at (0,-0.7) {\tiny $J(\mathcal{G}^*)=8$, CT=213s};
\end{scope}
\end{tikzpicture}
\\
\begin{tikzpicture}[xscale=1.4,yscale=1.4]
\tikzset{vertex/.style = {shape=circle,draw,minimum size=1pt, inner sep=0pt}}
\begin{scope}[edge/.style = {->,> = latex',line width=0.7pt},node distance=2pt]
\node[vertex] (1) at  (0,1.3) {$s_1$};
\node[vertex] (2) at  (-0.7,0.8) {$s_2$};
\node[vertex] (3) at  (-0.5,0) {$s_3$};
\node[vertex] (4) at  (0.5,0) {$s_4$};
\node[vertex] (5) at  (0.7,0.8) {$s_5$};
\draw[edge] (1) to (3);
\draw[edge] (1) to (4);
\draw[edge] (2) to (1);
\draw[edge] (3) to (5);
\draw[edge] (4) to (2);
\draw[edge] (4) to (3);
\draw[edge] (5) to (1);
\node at (0,-0.2) {\tiny (d) $K=4$};
\node at (0,-0.45) {\tiny $\eta=0.1,\epsilon=0.1$};
\node at (0,-0.7) {\tiny$J(\mathcal{G}^*)=7$, CT=721s};
\end{scope}
\end{tikzpicture}
\begin{tikzpicture}[xscale=1.4,yscale=1.5]
\tikzset{vertex/.style = {shape=circle,draw,minimum size=1pt, inner sep=0pt}}
\begin{scope}[edge/.style = {->,> = latex',line width=0.7pt},node distance=2pt]
\node[vertex] (1) at  (0,1.3) {$s_1$};
\node[vertex] (2) at  (-0.7,0.8) {$s_2$};
\node[vertex] (3) at  (-0.5,0) {$s_3$};
\node[vertex] (4) at  (0.5,0) {$s_4$};
\node[vertex] (5) at  (0.7,0.8) {$s_5$};
\draw[edge] (1) to (2);
\draw[edge] (2) to (5);
\draw[edge] (3) to (4);
\draw[edge] (4) to (1);
\draw[edge] (5) to (3);
\node at (0,-0.2) {\tiny (e) $K=4$};
\node at (0,-0.45) {\tiny $\eta=0.1,\epsilon=0.02$};
\node at (0,-0.7) {\tiny$J(\mathcal{G}^*)=5$, CT=226s};
\end{scope}
\end{tikzpicture}
\begin{tikzpicture}[xscale=1.4,yscale=1.5]
\tikzset{vertex/.style = {shape=circle,draw,minimum size=1pt, inner sep=0pt}}
\begin{scope}[edge/.style = {->,> = latex',line width=0.7pt},node distance=2pt]
\node[vertex] (1) at  (0,1.3) {$s_1$};
\node[vertex] (2) at  (-0.7,0.8) {$s_2$};
\node[vertex] (3) at  (-0.5,0) {$s_3$};
\node[vertex] (4) at  (0.5,0) {$s_4$};
\node[vertex] (5) at  (0.7,0.8) {$s_5$};
\node at (0,-0.2) {\tiny (f) $K=6$};
\node at (0,-0.45) {\tiny $\eta=0.1,\epsilon=0.01$};
\node at (0,-0.7) {\tiny$J(\mathcal{G}^*)=0$, CT=1s};
\end{scope}
\end{tikzpicture}
\caption{Optimal communication graphs for different values of $K$ (complexity of the controller), $\eta$ (maximum magnitude of the allowed disturbances) and $\epsilon$ (the degree of dynamical couplings). ``CT" stands for computation time.}
\label{fig_various}
\end{figure}
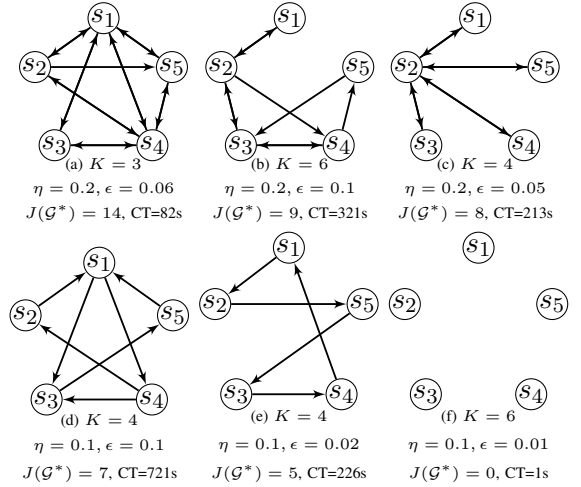

\begin{figure}
\vspace{1pt}
\centering
\includegraphics[width=0.22\textwidth]{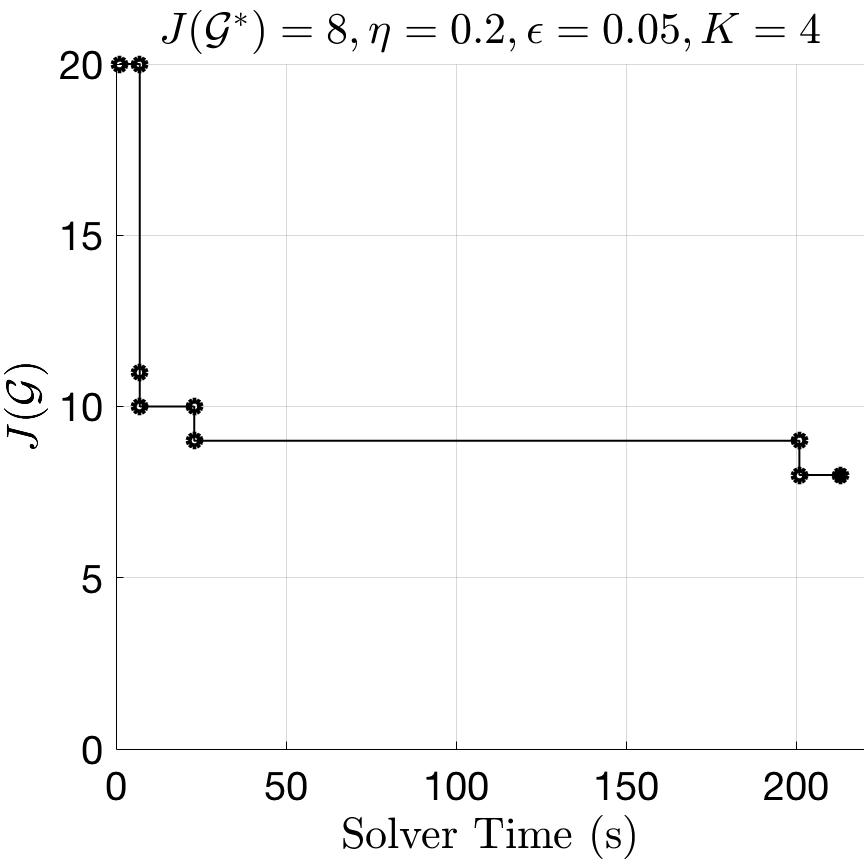}
\includegraphics[width=0.22\textwidth]{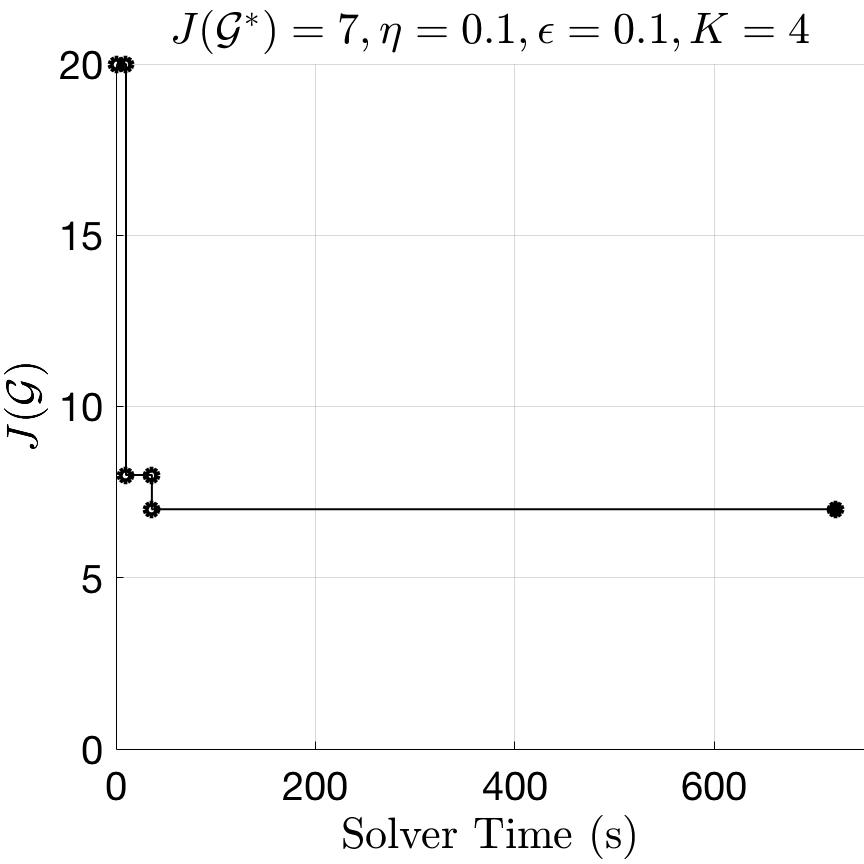}
\caption{The costs of incumbent feasible solutions versus time.}
\label{fig_solver}
\end{figure}

\subsection{Platooning}
We adopt a simplified version of the model in \cite{sadra2017platoon}. A platoon is a string of $N_p$ autonomous vehicles following a leader $l$. We have $N=N_p$ subsystems. System \eqref{eq_system} and sets \eqref{eq_sets} are constructed from what is described below. The state of each follower vehicle $s\in \mathcal{S}$ is $x_s=(d_s,v_s)$, where $d_s$ represents the distance from the preceding vehicle and $v_s$ is its velocity in the leader's frame. The evolution is given by:
\begin{equation}
\label{eq_platoon}
\begin{array}{l}
d_s[t+1]=d_s[t]-v_s[t]+v_{s'}[t]+\delta^x_s[t], \\
v_s[t+1]=v_s[t]+u_{s}[t]+\delta^v_s[t]+\delta^v_l[t], \\
\end{array}
\end{equation} 
where $s' \in \mathcal{S} \cup \{l\}$ is the preceding vehicle, $\delta^v_s[t] \in [-\varepsilon,\varepsilon], \delta^x_s[t] \in \frac{1}{10}[-\varepsilon,\varepsilon]$, are the disturbances hitting a follower vehicle, and $\delta^v_l[t] \in [-\varepsilon,\varepsilon]$ is the disturbance hitting the leader, which makes the frame non-inertial. We vary $\varepsilon$ in this example. Note that \eqref{eq_platoon} is a quite adversarial model since we consider independent disturbances affecting the distance evolution. The objective is to avoid rear-end avoid collisions for all times by writing $d_s[t] \ge -0.5, \forall t\ge 0,$ (a distance offset is performed in a way that $d_s<-0.5$ implies collision), and $\sum_\mathcal{S} d_s[t] \le \frac{1}{2} N_p, \forall t\in \mathbb{N}$ - the length of the platoon is always bounded. We also have bounded controls: $u_s[t] \in [-1,1], \forall s\in \mathcal{S}, \forall t\in \mathbb{N}$. System \eqref{eq_platoon} and specification are put into the form \eqref{eq_system}, and \eqref{eq_sets}, respectively. Note that the number of rows in $\mathbb{W}$ scale quadratically with the platoon size - $\mathbb{W}$ has a more complicated shape than a box \cite{sadra2017platoon}. 

\subsubsection{Structured Control}

We solve Problem 1. Consider a communication graph that every vehicle sends information to its follower (see Fig. \ref{fig_platoon} (b)). We set $\epsilon=0.05$. We observe that the minimum $K$ such that a feasible solution is found is $N_p+1$. The results are shown in Table. \ref{table_platoon_margin} for $K=N_p+1$. It is observed that the margin of correctness gradually decreases with the platoon size, highlighting the fundamental limits of predecessor following \cite{sabuau2017optimal}.   

\begin{table}[t]
\vspace{0.1in}
\centering
\caption{Margins of Correctness for Predecessor Following}
\resizebox{0.48\textwidth}{!}{%
\begin{tabular}{|c||c|c|c|c|c|c|c|c|c|}
\hline
$N_p$ & $3$ & $4$ & $5$ & $6$ & $8$ & $10$ & $12$ & $15$ \\
\hline
$\rho^*$ & $0.727$ & $0.726$ & $0.723$ & $0.721$ & $0.716$ & $0.710$ & $0.704$ & $0.697$ \\
\hline
CT(s) & $-$ & $-$ & $0.05$ & $0.1$ & $0.7$ & $3$ & $13$ & $133$ \\
\hline
\end{tabular}
}
\label{table_platoon_margin}
\end{table}

\subsubsection{Graph Design}
We solve Problem 2. 
We let $N=6$ and $c_{s_i,s_j}=(i-j)^2$ to penalize longer communication links. Some particular optimal graphs are shown in Fig. \ref{fig_platoon}. It is observed that for small disturbances, no communication is needed at all, but in order to attenuate heavier disturbances without violating collision and platoon length constraints, more communication links are required.   
 
\begin{figure}[t]
\centering
\vspace{5pt}
\begin{tikzpicture}[xscale=1.2,yscale=1]
\tikzset{vertex/.style = {minimum size=0pt}}
\begin{scope}[edge/.style = {->,> = latex',line width=0.7pt},node distance=1pt]
    \node[vertex](4) at (0,0) {\includegraphics[width=0.04\textwidth]{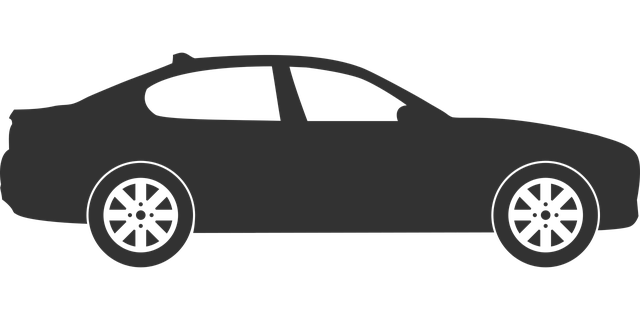}};
    \node[vertex](4) at (1,0) {\includegraphics[width=0.04\textwidth]{car}};
    \node[vertex](3) at (2,0) {\includegraphics[width=0.04\textwidth]{car}};
    \node[vertex](2) at (3,0) {\includegraphics[width=0.04\textwidth]{car}};
    \node[vertex](1) at (4,0) {\includegraphics[width=0.04\textwidth]{car}};
    \node[vertex](1) at (5,0) {\includegraphics[width=0.04\textwidth]{car}};
    \node[vertex](1) at (6,0) {\includegraphics[width=0.04\textwidth]{car}};
   
   	\node[] at (1,0.31) {\includegraphics[height=0.01\textwidth]{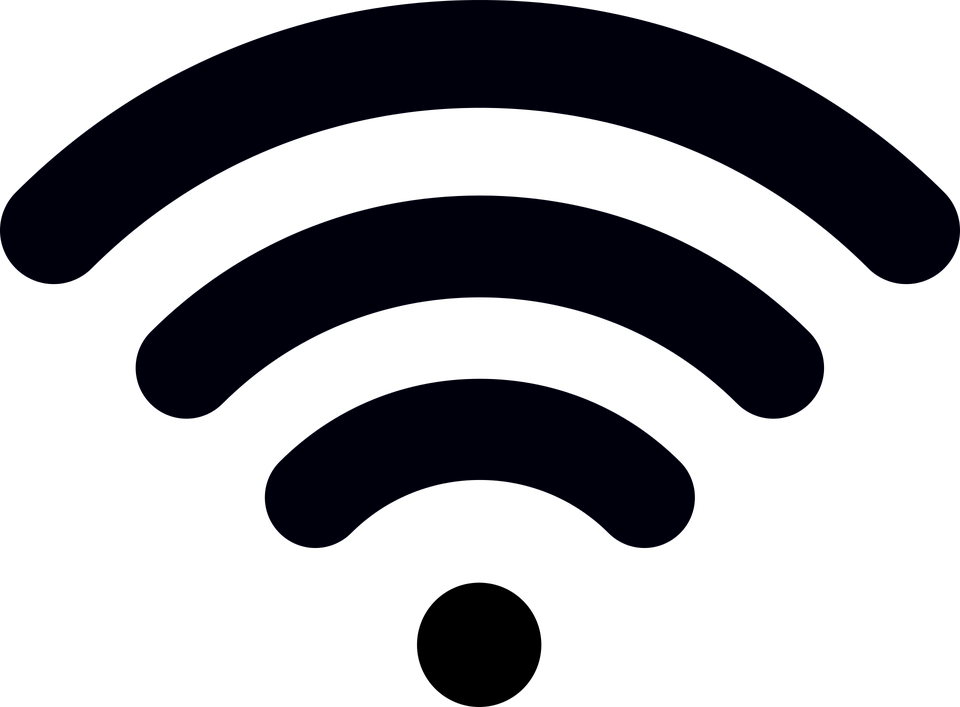}};
   	\node[] at (2,0.31) {\includegraphics[height=0.01\textwidth]{wifi}};
   	\node[] at (3,0.31) {\includegraphics[height=0.01\textwidth]{wifi}};
	\node[] at (0,0.31) {\includegraphics[height=0.01\textwidth]{wifi}};
   	\node[] at (5,0.31) {\includegraphics[height=0.01\textwidth]{wifi}};
   	\node[] at (4,0.31) {\includegraphics[height=0.01\textwidth]{wifi}};
   	\node[] at (6,0.45) {$l$};
    \draw[]  (-0.3, -0.11 ) -- (6.3, -0.11 );

\node[] at (3,-0.45) {\tiny (a) $\varepsilon=0.15,K=8, J(\mathcal{G}^*)=0$, CT=7s};
\end{scope}
\end{tikzpicture}
\begin{tikzpicture}[xscale=1.2,yscale=1]
\tikzset{vertex/.style = {minimum size=0pt}}
\begin{scope}[edge/.style = {->,> = latex',line width=0.7pt},node distance=1pt]
    \node[vertex](4) at (0,0) {\includegraphics[width=0.04\textwidth]{car}};
    \node[vertex](4) at (1,0) {\includegraphics[width=0.04\textwidth]{car}};
    \node[vertex](3) at (2,0) {\includegraphics[width=0.04\textwidth]{car}};
    \node[vertex](2) at (3,0) {\includegraphics[width=0.04\textwidth]{car}};
    \node[vertex](1) at (4,0) {\includegraphics[width=0.04\textwidth]{car}};
    \node[vertex](1) at (5,0) {\includegraphics[width=0.04\textwidth]{car}};
    \node[vertex](1) at (6,0) {\includegraphics[width=0.04\textwidth]{car}};
   
   	\node[] at (1,0.31) {\includegraphics[height=0.01\textwidth]{wifi}};
   	\node[] at (2,0.31) {\includegraphics[height=0.01\textwidth]{wifi}};
   	\node[] at (3,0.31) {\includegraphics[height=0.01\textwidth]{wifi}};
	\node[] at (0,0.31) {\includegraphics[height=0.01\textwidth]{wifi}};
   	\node[] at (5,0.31) {\includegraphics[height=0.01\textwidth]{wifi}};
   	\node[] at (4,0.31) {\includegraphics[height=0.01\textwidth]{wifi}};
   	\node[] at (6,0.45) {$l$};
    \draw[]  (-0.3, -0.11 ) -- (6.3, -0.11 );
\draw [thick,->] (0.9,0.4) to[->,in=45,out=135] (0.1,0.4);
\draw [thick,->] (2.9,0.4) to[->,in=45,out=135] (2.1,0.4);
\draw [thick,->] (1.9,0.4) to[->,in=45,out=135] (1.1,0.4);
\draw [thick,->] (3.9,0.4) to[->,in=45,out=135] (3.1,0.4);
\draw [thick,->] (4.9,0.4) to[->,in=45,out=135] (4.1,0.4);
\node[] at (3,-0.45) {\tiny (b) $\varepsilon=0.1800,K=8, J(\mathcal{G}^*)=5$, CT=81s};
\end{scope}
\end{tikzpicture}

\begin{tikzpicture}[xscale=1.2,yscale=1]
\tikzset{vertex/.style = {minimum size=0pt}}
\begin{scope}[edge/.style = {->,> = latex',line width=0.7pt},node distance=1pt]
    \node[vertex](4) at (0,0) {\includegraphics[width=0.04\textwidth]{car}};
    \node[vertex](4) at (1,0) {\includegraphics[width=0.04\textwidth]{car}};
    \node[vertex](3) at (2,0) {\includegraphics[width=0.04\textwidth]{car}};
    \node[vertex](2) at (3,0) {\includegraphics[width=0.04\textwidth]{car}};
    \node[vertex](1) at (4,0) {\includegraphics[width=0.04\textwidth]{car}};
    \node[vertex](1) at (5,0) {\includegraphics[width=0.04\textwidth]{car}};
    \node[vertex](1) at (6,0) {\includegraphics[width=0.04\textwidth]{car}};
   
   	\node[] at (1,0.31) {\includegraphics[height=0.01\textwidth]{wifi}};
   	\node[] at (2,0.31) {\includegraphics[height=0.01\textwidth]{wifi}};
   	\node[] at (3,0.31) {\includegraphics[height=0.01\textwidth]{wifi}};
	\node[] at (0,0.31) {\includegraphics[height=0.01\textwidth]{wifi}};
   	\node[] at (5,0.31) {\includegraphics[height=0.01\textwidth]{wifi}};
   	\node[] at (4,0.31) {\includegraphics[height=0.01\textwidth]{wifi}};
   	\node[] at (6,0.45) {$l$};
    \draw[]  (-0.3, -0.11 ) -- (6.3, -0.11 );
\draw [thick,->] (0.9,0.4) to[->,in=45,out=135] (0.1,0.4);
\draw [thick,->] (2.9,0.4) to[->,in=45,out=135] (2.1,0.4);
\draw [thick,->] (1.9,0.4) to[->,in=45,out=135] (1.1,0.4);
\draw [thick,->] (3.9,0.4) to[->,in=45,out=135] (3.1,0.4);
\draw [thick,->] (4.9,0.4) to[->,in=45,out=135] (4.1,0.4);

\draw [thick,->] (4.9,0.4) to[->,in=45,out=135] (3.1,0.4);

%\draw [thick,->] (3.9,0.4) to[->,in=45,out=135] (2.1,0.4);

\draw [thick,->] (2.9,0.4) to[->,in=45,out=135] (1.1,0.4);
\draw [thick,->] (2.9,0.4) to[->,in=45,out=135] (0.1,0.4);

\draw [thick,->] (1.9,0.4) to[->,in=45,out=135] (0.1,0.4);

\node[] at (3,-0.45) {\tiny (c) $\varepsilon=0.1836,K=8, J(\mathcal{G}^*)=26$, CT=32s};
\end{scope}
\end{tikzpicture}
\caption{Optimal communication graphs for platooning. The leader is considered as an adversary with bounded acceleration range. }
\label{fig_platoon}
\end{figure}

\section{Concluding Remarks and Future Work}
We introduced a class of distributed control policies for networked linear systems subject to polytopic constraints and disturbances. We both explored designing optimal control policies and optimal communication graphs.
The key idea was taking the convex parameterization of RCI sets from \cite{rakovic2007optimized}, and transforming the memoryless piecewise affine laws to linear laws with memory. 
We applied our method to systems with couplings in both dynamics and constraints.
Future work will investigate the limits of our approach and exploring possibly more general classes of distributed set-invariance controllers.   

\section*{Acknowledgement}
We thank MirSaleh Bahavarnia from Lehigh University for fruitful discussions during preparation of this manuscript. 

\balance
\bibliography{ana_references}

\end{document}